\documentclass[]{ec-acmsmall} 

\usepackage[numbers]{natbib}

\usepackage{amsmath,amsfonts,graphpap,amscd,mathrsfs,graphicx,lscape}
\usepackage{epsfig,amssymb,amstext,xspace}
\usepackage{subfigure}

\usepackage{color}              
\usepackage{epsfig}

\usepackage[suppress]{color-edits}
\addauthor{vs}{red}
\addauthor{bl}{blue}

\usepackage[ruled,vlined]{algorithm2e}
\renewcommand{\algorithmcfname}{ALGORITHM}
\SetAlFnt{\small}
\SetAlCapFnt{\small}
\SetAlCapNameFnt{\small}
\SetAlCapHSkip{0pt}
\IncMargin{-\parindent}

\newcommand{\E}{\mathbb{E}}

\newcommand{\rev}{\mathcal{R}}

\newtheorem{defn}[theorem]{Definition}
\newtheorem{observation}[theorem]{Observation}
\newtheorem{conj}[theorem]{Conjecture}

\def\squareforqed{\hbox{\rule{2.5mm}{2.5mm}}}
\def\blackslug{\rule{2.5mm}{2.5mm}}
\def\qed{\hfill\blackslug}

\def\QED{\ifmmode\squareforqed 
  \else{\nobreak\hfil   
    \penalty50                 
    \hskip1em                  
    \null                      
    \nobreak                   
    \hfil                      
    \squareforqed              
    \parfillskip=0pt           
    \finalhyphendemerits=0     
    \endgraf}                  
  \fi}

\def\blksquare{\rule{2mm}{2mm}}
\def\qedsymbol{\blksquare}
\newcommand{\bg}[1]{\medskip\noindent{\bf #1}}
\newcommand{\ed}{{\hfill\qedsymbol}\medskip}

\newenvironment{proofof}[1]{{\it{Proof of #1 : }}}{\ed}

\newcommand{\R}{\ensuremath{\mathbb R}}

\newcommand{\I}{\ensuremath{\mathcal I}}

\newcommand{\El}{\ensuremath{\mathcal{L}}}

\newcommand{\comment}[1]{}
 {}

\newcommand{\junk}[1]{}

\newlength{\tmp} \newlength{\lpsx} \newlength{\lpsy} \newlength{\upsx} \newlength{\upsy}

\newcommand{\M}{\ensuremath{\mathcal M}}

\newcommand{\X}{\ensuremath{\mathcal X}}

\newcommand{\x}{\ensuremath{x}}
\newcommand{\xsi}{\ensuremath{x_i}}

\newcommand{\V}{\ensuremath{\mathcal V}}

\newcommand{\val}{\ensuremath{v}}

\newcommand{\vali}{\ensuremath{v_i}}

\newcommand{\A}{\ensuremath{\mathcal A}}
\newcommand{\Ai}{\ensuremath{\mathcal A}_i}

\newcommand{\ac}{\ensuremath{a}}

\newcommand{\al}{\ensuremath{\mathbf a}}

\newcommand{\xos}{\text{\textsc{XOS}}}

\newcommand{\ground}{\ensuremath{\mathcal{E}}}
\newcommand{\matroid}{\ensuremath{\textsc{M}}}

\newcommand{\poa}{\text{\textsc{PoA}} }

\newcommand{\opt}{\text{\textsc{Opt}} }

\newcommand{\OurTitle}{Greedy Algorithms make Efficient Mechanisms}

\makeatletter
\def\runningfoot{\def\@runningfoot{}}
\def\firstfoot{\def\@firstfoot{}}
\makeatother

\begin{document}

\markboth{B. Lucier and V. Syrgkanis}{\OurTitle}

\title{\OurTitle}

\author{Brendan Lucier\affil{Microsoft Research; \texttt{brlucier@microsoft.com}}
Vasilis Syrgkanis\affil{Microsoft Research; \texttt{vasy@microsoft.com}}
}
\date{}

\begin{abstract}
We study mechanisms that use greedy allocation rules and pay-your-bid pricing to allocate resources subject to a matroid constraint.
We show that all such mechanisms obtain a constant fraction of the optimal welfare at any equilibrium of bidder behavior, via a smoothness argument. This unifies numerous recent results on the price of anarchy of simple auctions. Our results extend to polymatroid and matching constraints, and we discuss extensions to more general matroid intersections.
\end{abstract}



\maketitle

%

\section{Introduction}
A principle is tasked with allocating resources to a set of self-interested agents, subject to constraints that determine which allocations are feasible.  The principle (or ``seller'') wishes to allocate resources in a socially efficient manner, with the understanding that the agents will interact with any prescribed mechanism in such a way to maximize their own utilities.  This familiar setup, now ubiquitous in the algorithmic game theory literature, captures a variety of important scenarios: from selling a single good on eBay, to scheduling tasks in a cloud-based system, to running auctions for online advertising space.  A classic approach to this mechanism design problem is to elicit, from each agent, a full description of their preferences over all possible outcomes, then find the welfare-maximizing allocation according to the bids and charge appropriate (truth-incentivizing) payments.  Unfortunately, this optimal mechanism scales poorly as the number of potential outcomes and agents grows, both in the difficulty of the optimization task and in the communication requirements imposed by the protocol.

There has been a significant recent line of work studying the performance of so-called ``simple'' mechanisms.  These auctions are typically characterized by restricted bidding languages, simple greedy allocation rules, and straightforward payment computations -- practical considerations that can be more desireable than optimality or truthfulness.  Examples include the GSP auction for sponsored search ads \cite{Caragiannis2014}, simultaneous item auctions for selling indivisible items to buyers with submodular valuations \cite{Christodoulou2008,Hassidim2011}, uniform-price auctions for selling many homogeneous goods \cite{Markakis2012},  position auctions with externalities \cite{RT12} and many others.  These mechanisms are not truthful nor efficient, but each has been shown to achieve a constant fraction of the optimal welfare at every Nash equilibrium of bidder behavior.  

Recent works \cite{Syrgkanis2013,Syrgkanis2012,Roughgarden2012} showed that each of these mechanisms satisfies a natural best-response property, termed \emph{smoothness}, and that this property implies approximate efficiency at equilibrium.  Moreover, smoothness not only guarantees high efficiency, but this guarantee is robust in many ways: it extends to learning behavior on the part of the bidders, uncertainty on the parameters of the game, simultaneous and sequentially occurring mechanisms, and budget constraints. On the other hand, truthful mechanisms tend to not have such \emph{composability} properties; one might go so far as to say that incentive compatible mechanisms \emph{overfit} the mechanism design setting by treating it in isolation and not taking into account strategic influences and variations external to the mechanism.

The smoothness property, while eminently useful, is a semantic concept.  It is based on an existential property of a given mechanism, and does not directly give algorithmic guidelines about what mechanisms are smooth or how one should design them.  This is reminiscient of truthfulness, which is also a semantic property; it becomes far more useful and usable when paired with a descriptive algorithmic condition, such as optimality (as in the VCG mechanism) or monotonicity (for single-parameter problems).  
\emph{Can we give analogous, useful characterizations of algorithmic conditions that guarantee smoothness?}

A common feature in each of the smooth mechanisms described above is the greediness of the allocation rule.  Indeed, an intuition that arises from this line of work is that greedy algorithms lend themselves well to auction design, in the sense that they generate smooth mechanisms that necessarily have good performance at equilibrium.
However, to this point, there is no result of satisfying generality that links greediness with smoothness.  The arguments establishing smoothness for each mechanism are subtley tailored to their specific contexts, despite the intuition of a more general principle.
%
The contribution of this work is to formalize this intuition and unify the aforementioned results; to show that, in a general sense, greedy auctions are smooth.
\vsedit{For instance, we unify results on the efficiency of simultaneous item auctions \cite{Christodoulou2008,Hassidim2011}, uniform price auction \cite{Markakis2012} and position auctions with externalities \cite{RT12} in a single argument (albeit with slightly worse constants).}

Specifically, we show that if a greedy allocation rule is used to allocate resources subject to a matroid constraint, and buyers have submodular\footnote{Our results actually hold for the more general class of fractionally subadditive preferences.} preferences over the resources, then the resulting mechanism is smooth and will achieve a constant fraction of the optimal welfare at Nash equilibrium.  In other words, the mechanism has constant price of anarchy.
We then show how each of the examples described above, plus others, fall within this framework.  Unsurprisingly, the constants we obtain for the general results are not as tight as the constants achieved via direct analysis of specific auction instances.  Nevertheless, we view our general result as providing insight into the structural properties of mechanisms that lead to reasonable approximation guarantees.



\paragraph{Challenges and Techniques}

To illustrate the barriers to a general connection between greediness and smoothness, let us first describe a standard smoothness argument typical of the recent literature on simple auctions.  Consider the following very simple setting: the sale of a single item via first-price, sealed-bid auction.  This is a greedy allocation rule subject to the (very simple) matroid constraint that at most one bid can be chosen.  Suppose agent $i$ has the maximum value $v_i$ for the item, and let $B = (b_1, \dotsc, b_n)$ be a bid profile at equilibrium.  We wish to argue that the welfare generated at this equilibrium is close to $OPT$, which is $v_i$.  To see this, imagine what would happen if agent $i$ deviated to bidding half of his true value, $v_i/2$.  Agent $i$ either wins under this deviation, or he does not.  If he does, then his utility would be $v_i/2$.  If he does not, then there must exist some other agent $j$ that bids more than $v_i/2$, and hence pays at least $v_i/2$ at equilibrium.  Either way, we can conclude that agent $i$'s utility under the deviation, plus the sum of all agents' payments at equilibrium, is at least $v_i/2$.  This is precisely the smoothness condition.  To deduce a price of anarchy bound, note that since bid profile $B$ forms an equilibrium, agent $i$'s utility at equilibrium must be at least his utility under the deviation.  So we can conclude that the sum of agent utilities \emph{at equilibrium}, plus the sum of payments at equilibrium, is at least $v_i/2$.  Since the two sums add up to the social welfare at equilibrium, we conclude that the welfare at equilibrium is at least $v_i/2$, half of the optimal welfare.

The smoothness condition is effectively a charging argument.  If agent $i$'s utility at equilibrium is low, relative to his welfare in the optimal outcome, then we aim to find an agent $j$ against whose payment we can charge agent $i$'s loss in utility.  In many cases, like the one above, the agent to charge against is obvious: for example, the winner of an item that was ``supposed'' to go to $i$ in the optimal allocation.  However, more generally, finding an appropriate charging method is not as straightforward, as our next example illustrates.  One technical challenge in extending the argument to arbitrary matroids is to show how to construct such mappings in general.

\begin{figure}[t]
\centering
\subfigure[example with a graphical matroid constraint.]
{\label{f:graphical}%
{
\includegraphics[width=1.3in]{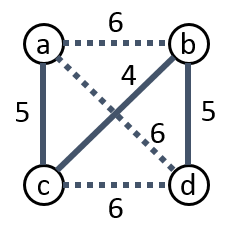}
}
}
\caption{Sample input instance to a greedy mechanism.  Elements are edges, and weights represent agent values.  Each solid edge is being bid on truthfully, each dashed edge has a bid of $0$.  The feasibility constraint is a graphical matroid (allocation cannot contain cycles).}
\label{fig:example}
\end{figure}

Consider the example in Figure \ref{f:graphical}, where elements are the edges in a graph and the feasibility constraint is the graphical matroid; that is, a set is feasible if it contains no cycles.  Each edge is labeled with a value, and is being bid on by a different agent.  The optimal outcome is $\{ab,ad,cd\}$, but the bid profile (where dashed lines get bid $0$) results in outcome $\{ac,bc,bd\}$.  If the agent bidding on edge $ab$ raised his bid to half of its value, $3$, he would lose this bid; to which equilibrium payment should this loss be charged?  A first instinct would be to charge against edge $ad$, since that edge sets the critical threshold for whether a bid on $ab$ would win.  However, such a charging scheme would also charge losing bids for edges $ad$ and $cd$ to edge $ac$ as well, leading to poor smoothness bounds.  To avoid such collisions, one must take a more global view, taking into account all deviations by all players.  In this example, one could charge $ab$ to $ac$, $cd$ to $bd$, and $ad$ to $bc$.  To build such a global charging method for general matroids, we use the Rota Exchange property of matroids, which guarantees the existence of an appropriate mapping on which to build a charging scheme.

%

\paragraph{Our Results}

Our main result is that any greedy mechanism that allocates resources subject to a matroid constraint has price of anarchy at most $3$ via a smoothness argument.

\bigskip
\noindent
\textbf{Main Theorem.} {\em
Suppose agents desire subsets of a ground set, which can be allocated subject to a matroid constraint.  Suppose further that agent valuations over elements are fractionally subadditive.  Then a mechanism that elicits separate bids on the ground set elements, then allocates greedily by bid and charges pay-your-bid payments, has Price of Anarchy at most $3$.
}\bigskip


Roughly speaking, a mechanism is smooth if, for every valuation profile $v$, there is an action profile $a^*$ such that, for all bid profiles $b$, the sum of all agent payments (under $b$) plus the sum of agent utilities (under unilateral deviations from $b$ to $a^*$) is at least a constant factor of the optimal social welfare. To generalize previous smoothness arguments to arbitrary matroid environments, we must find a general method to charge losing bids (under deviation $a^*$) to agent payments (under $b$), as described above.  
We make use of the well-known Rota exchange property of matroids to build the necessary mappings.  Mapping in hand, we employ a method of declared-welfare maximization recently developed in \cite{Babaioff2014}, paired with the composition framework of \cite{Syrgkanis2013}, to establish the claimed smoothness condition.


We next show that our result can be extended to handle polymatroid constraints.  A formal definition of polymatroids appears in Section \ref{sec.polymatroid}.  Roughly speaking, 
a polymatroid 
allows fractional allocations of goods, subject to submodular constraints on the total allocated quantities of different subsets of goods.  
We actually provide two different forms of this extension.  First, if each agent's value for a good is linear in the amount of the good she receives, then we consider a greedy mechanism that asks agents to report their per-unit value of each good, then allocates the goods greedily.  
%
Second, we consider more general 
valuations that can be arbitrary submodular functions over the space of feasible allocations; in particular, they needn't be linear for a single good.  For this setting, we modify our mechanism to extend the bidding language slightly: it elicits from each agent a bid on each element of the ground set, where the bid includes a declared value per unit of the element, plus a maximum amount of the element that the agent is willing to accept.  One can then consider the greedy allocation rule, which maximizes resources greedily subject to the polymatroid contraint, up to the declared caps.  
For both extensions, our analysis applies and we show that the greedy mechanism has a price of anarchy of at most $3$.

Our main theorem and extensions can be applied to a variety of auction settings, such as sponsored search auctions and combinatorial auctions.  We discuss these applications in Section \ref{sec:applications}.

We then turn to matching constraints, where the elements for sale are edges in a bipartite graph and an allocation is feasible only if no two allocated edges are coincident.  This is a special case of the intersection of two matroids.  For this case we again show a constant price of anarchy via smoothness analysis, though this proof is more technically involved and we obtain a bound of $8$.\footnote{
The technical difficulties we encounter bear some similarities to the independent work of Kesselheim et al. \cite{Kesselheim2014} who  consider an auction setting that involves matroid constraints. They consider a greedy mechanism on the intersection of a “unit”-partition matroid and a separate matroid constraint for each player.  This result is incomparable to ours, as we provide bounds for greedy algorithms on the intersection of two “unit”-partition matroids, i.e. matchings.
}

We conjecture that our results extend to the intersection of two arbitrary matroids, and perhaps even to the intersection of $k$ matroids (with a price of anarchy bound of $O(k)$).  As progress toward understanding matroid intersections, we show that if the feasibility constraint is the intersection of $k$ matroids and the \emph{optimal} allocation rule is used on the bids (rather than the greedy rule), then the price of anarchy is indeed $O(k)$.  We emphasize here that this analysis is very different from that of different ``optimal'' allocation methods, such as the VCG mechanism.  In particular, the allocation rule is being applied to a restricted bidding space, where agents must place separate bids on each item; the allocation returned is the optimal allocation \emph{for the bids}, which necessarily describe only additive valuations.

\paragraph{Other Related Work}

Our results are similar in spirit to Borodin and Lucier \cite{Lucier2010}, who also consider the price of anarchy of greedy mechanisms.  They show that for a class of combinatorial auction problems, any $c$-approximate greedy algorithm has price of anarchy $O(c)$ when paired with a first-price payment rule.  Our results are incomparable to theirs: they do not impose requirements on the feasibility constraints or the agents' valuation classes, but they allow only greedy algorithms in which agents bid directly on their outcomes (e.g., a package of items in a combinatorial auction).  Since our analysis doesn't impose this last requirement, it applies to a wider array of greedy mechanisms, such as combinatorial auctions that allocate individual items or generalized sponsored search auctions that don't require customers to bid directly on slots.

%
%

Babaioff, Lucier, Nisan, and Paes Leme \cite{Babaioff2014} give a general price of anarchy analysis for ``declared welfare maximizing'' mechanisms, which collect bids (from a potentially limited bidding space) and then allocate optimally with respect to the bids.  Our greedy mechanism for matroids is a declared welfare maximizer, but our analysis applies to a more general space of feasibility conditions; the analysis of \cite{Babaioff2014} is specific to combinatorial auction constraints.  That said, portions of our proofs apply techniques in a spirit similar to \cite{Babaioff2014}.

Bikhchandani et al. \cite{Bikhchandani2011} design an ascending-price Vickrey-style auction for selling elements of a matroid.  Their setting is identical to our own, except that it is assumed that agent values are additive over the elements they receive.  They show that their mechanism is welfare-optimal and has a variety of desireable incentive properties.  Our work extends their setting to the case of more general valuations, where greedy allocation rules are not necessarily incentive compatible.

\section{Preliminaries}

A principal wants to decide an allocation of resources to $n$ players. The principal has to choose an allocation $x=(x_1,\ldots,x_n)$ among a set of feasible allocation vectors $\X$ that is a subset of a product space of allocations $\X\subseteq \X_1\times \ldots\times \X_n$. We assume that the principal can also ask the players to pay for their allocation and thereby he also needs to decide a payment vector $p=(p_1,\ldots,p_n)\in \R_+^n$. 

Each player $i$, has a valuation function that maps an allocation to some non-negative real number: $v_i:\X_i\rightarrow \R_+$. We will denote with $\V_i$ the set of allowed valuations for player $i$ and with $\V=\V_1\times\ldots\times \V_n$ the set of allowed valuation profiles. If a player is given allocation $x_i$ and is asked to pay $p_i$, then her utility is quasi-linear  with respect to money: 
$u_i(\xsi, p_i; \vali) = \vali(\xsi) - p_i.$

\begin{defn}[Mechanism] 
A mechanism $\M$ is a tuple $(\A,X,P)$, where $\A=\A_1\times\ldots\times \A_n$ and $\Ai$ is a set of actions available to player $i$, $X:\A\rightarrow \Delta(\X)$ is an allocation function that maps each action profile $\ac=(\ac_1,\ldots,\ac_n)$ to a distribution over feasible allocation vectors 
and $P:\A\rightarrow \Delta(\R^n_+)$ is a payment function that maps each action profile to a distribution over payment vectors. 
\end{defn}

We will denote with $X_i$ and $P_i$ the $i$-th coordinate of the allocation and payment functions respectively and with $U_i^{\M}: \Delta(\A)\times \V_i \rightarrow \R$ the expected utility of player $i$ from mechanism $\M$.  That is,
$U_i^{\M}(\al; v_i) =  \E_{\al, X_i(\al), P_i(\al)}\left[v_i\left(X_i(\al)\right)-P_i\left(\al\right)\right]$.
We will also denote the expected revenue of the mechanism as
$\rev^\M(\al) = \sum_{i\in [n]} \E_{\al,P_i(\al)}[P_i(\al)]$.
\paragraph{Efficiency} We will denote with $SW:\A\times \V \rightarrow \R_+$, the social welfare produced by the mechanism under some action profile:
$SW^{\M}(a; v)=\sum_{i\in [n]} \E_{X_i(a)}[\val_i\left(X_i(a)\right)].$
For any valuation profile $\val\in \V$ we will denote with $\x^*(v)$ the optimal allocation, i.e. the allocation that maximizes welfare over all feasible allocations $\x\in \X$ and we will denote with
$\opt(\val)= \sum_{i\in [n]} \val_i(\x_i^*(\val))$.

\begin{defn}[Smooth Mechanism \cite{Syrgkanis2013}]\label{def:smooth-mech}
A mechanism $\M$ is $(\lambda,\mu)$-smooth for some $\lambda,\mu\geq 0$, if for any valuation profile $v\in \V$ and for each player $i\in [n]$ there exists a randomized action $\al_i^*(v)$, such that for any action profile $a\in \A$:
\begin{equation}
\textstyle{\sum_{i\in [n]} U_i^\M(\al_i^*(v),a_{-i}; v_i)\geq \lambda\cdot \opt(v)- \mu\cdot \rev^\M(a)}
\end{equation}
\end{defn}

\begin{defn}[Smooth Mechanism via Swap Deviations]\label{defn:swap-smooth}
A mechanism $\M$ is $(\lambda,\mu)$-smooth via swap deviations if for any valuation profile $v\in \V$, there exists a mapping $\al_i^*(v,\cdot): \A_i\rightarrow \Delta(\A_i)$, such that for any action profile $a\in \A$:
\begin{equation}
\textstyle{\sum_{i\in [n]} U_i^\M(\al_i^*(v,a_i),a_{-i}; v_i)\geq \lambda \opt(v)- \mu \rev^\M(a)}
\end{equation}
\end{defn}

We refer the reader to \cite{Syrgkanis2013} for properties of smooth mechanisms. Roughly if a mechanism is $(\lambda,\mu)$-smooth, then every equilibrium has social welfare at least $\frac{\lambda}{\max\{1,\mu\}}$ of the optimal (i.e. the price of anarchy ($\poa$) is at most $\frac{\max\{1,\mu\}}{\lambda}$) and this extends to no-regret learning outcomes, incomplete information, simultaneous and sequentially occurring mechanisms and budget constraints.

\paragraph{Combinatorial Allocation Spaces and Greedy Mechanisms.}
Consider the following instantiation of the mechanism design setting: the allocation space $\X_i$ of each player $i$ consists of the power set of a finite set of elements $\ground_i$ which we will refer to as the ground elements of player $i$.  Hence, each bidder $i$'s valuation is a set function $v_i:2^{\ground_i}\rightarrow \R_+$ which maps an allocation $S_i\subseteq \ground_i$ to a value $v_i(S_i)$. We will be primarily interested in the case where this function is $\xos$, i.e. can be expressed as a maximum over a set of additive valuations (see Lehmann et al. \cite{Lehmann2001}), which is a superset of submodular valuations.

We denote with $\ground=\ground_1\cup \ldots \cup \ground_n$ the set of all ground elements. The outcome of the mechanism is a subset $S\subseteq \ground$ of this ground set and thereby the allocation of each player is $S_i= S\cup \ground_i$. We assume that there is some feasibility constraint $F\subseteq 2^\ground$ defined on $\ground$, which defines which subsets $S\subseteq \ground$ are feasible. The outcome of the mechanism is restricted to fall within $F$. 

\paragraph{Greedy Mechanism on Reported Bids} We consider the following mechanism: from each player $i$, the auctioneer solicits bids $b_t$ for each $t\in \ground_i$, i.e. $\A_i=\R_+^{|\ground_i|}$. We will denote with $a_i=(b_t)_{t\in \ground_i}$ an action of player $i$ and with $b=(b_t)_{t\in \ground}$ a bid profile on all elements. The auctioneer runs the greedy algorithm on the reported bid profile to decide which elements of $\ground$ are to be picked, i.e. elements are considered in decreasing ordered of bids and each element is added to the outcome as long as it is feasible. Each player is asked to pay his bid for each of his elements in his allocation.

\section{Smoothness for Matroid Feasibility Constraints}

We now proceed to the main result of the paper. Consider the case where the feasibility constraint $F$, is the collection $\I$ of independent sets of a matroid $\matroid=(\ground,\I)$ (see Schrijver \cite{Schrijver2003} for an extensive exposition of matroids), defined on the ground set. We show that the greedy mechanism on reported bids is a $\left(\frac{1}{3},1\right)$-smooth mechanism. 
To show the smoothness property we will heavily use an exchange property of matroid feasibility constraints, proved by Lee et al. \cite{Lee2010}.
\begin{lemma}[Generalized Rota Exchange \cite{Lee2010}]\label{lem:exchange} Let $\matroid = (\ground, I)$ be a matroid and $A, B\in \I$. Let $A_1, . . . , A_n$ be subsets of $A$ such that each element of $A$ appears in exactly $q$ of them. Then there are sets $B_1, . . . , B_m \subseteq  B$ such that each
element of $B$ appears in at most $q$ of them, and for each $i$, $A_i\cup (B/\ B_i) \in \I$.
\end{lemma}

\begin{theorem}
\label{thm.xos}
The greedy mechanism on reported bids is a $(\frac{1}{3},1)$-smooth mechanism when the valuation of each player is $\xos$. Hence, every equilibrium achieves at least $1/3$ of the optimal welfare. 
\end{theorem}
We will use the following Lemma, which is a re-interpretation of the results of \cite{Syrgkanis2013}):
\begin{lemma}[Syrgkanis and Tardos \cite{Syrgkanis2013}]\label{lem:sum-to-max} If a mechanism is $(\lambda,\mu)$-smooth for a class of valuations $\V$, then the same mechanism is also $(\lambda,\mu)$-smooth for the class of valuations $\max-\V$, defined as all valuations $v$, that can be expressed as the maximum over a set of valuations from $\V$: $v(S)=\max_{\ell\in \El} v^{\ell}(S)$, for some arbitrary potentially infinite index set $\El$.
\end{lemma}
Since $\xos$ are by definition valuations that can be expressed as a maximum over a set of additive valuations, it suffices to show that the mechanism is smooth for additive valuations, which we show in the following Lemma and which is the main technical part of the theorem.

\begin{lemma}\label{thm:poa-greedy}
The greedy mechanism on reported bids is a $(\frac{1}{3},1)$-smooth mechanism when the valuation of each player is additive.
\end{lemma}
\begin{proof}
Consider an additive valuation profile $v$, i.e. $v_i(S_i)=\sum_{t\in S_i}w_t$.  Let $S^*$ be the optimal base for valuation profile $v$ and $S_i^*=S^*\cap \ground_i$, be player $i$'s allocation in the optimal base. Suppose that each player $i$ deviates to $a_i^*=\left(\frac{w_t}{\alpha}\right)_{t\in \ground_i}$.

Consider an action profile $a$, where $a_i=(b_t)_{t\in \ground_i}$, and let $S$, be the selected set under action profile $a$. Let $a'=(a_i^*,a_{-i})$, be the induced action profile and $S'$ be the set allocated after the deviation and $S_i'=S'\cap \ground_i$. 

We denote with $W(S,a)=\sum_{t\in S} b_t$, the. By Lemma \ref{lem:exchange} for $q=1$, we have that  there exist disjoint sets $T_1,\ldots, T_n$ of $S$, such that $Q=S_i^* \cup (S - T_i)\in \I$. By optimality of the greedy algorithm on the reported bid profile and since $Q$ is feasible, we have: 
\begin{align*}
 \sum_{t\in S_i'}\frac{w_t}{\alpha} + \sum_{t\in S'-S_i'}b_t\geq~& \sum_{t\in S_i^*}\frac{w_t}{\alpha} + \sum_{t\in S-T_i- \ground_i}b_t + \sum_{t\in (S\cap\ground_i) -T_i}\frac{w_t}{\alpha}
\geq~ \sum_{t\in S_i^*}\frac{w_t}{\alpha} + \sum_{t\in S-T_i-\ground_i}b_t
\end{align*}
By optimality of the algorithm on the initial bid profile we have:
\begin{equation*}
\textstyle{\sum_{t\in S}b_t=W(S,a) \geq W(S',a) = \sum_{t\in S'}b_t \geq \sum_{t\in S'-S_i'}b_t}
\end{equation*}
Combining we get:
\begin{align*}
 \sum_{t\in S_i'}\frac{w_t}{\alpha} \geq~& \sum_{t\in S_i^*}\frac{w_t}{\alpha} + \sum_{t\in S-T_i- \ground_i}b_t -  \sum_{t\in S}b_t
=~\sum_{t\in S_i^*}\frac{w_t}{\alpha} - \sum_{t\in S\cap (T_i \cup \ground_i)}b_t \\
\geq~&\sum_{t\in S_i^*}\frac{w_t}{\alpha} - \sum_{t\in S\cap T_i }b_t - \sum_{t\in S\cap \ground_i}b_t
\end{align*}
Observe that by definition the utility of the player under the deviation is: $U_i^\M(a_i^*,a_{-i}; v_i) = \left(1-\frac{1}{\alpha}\right)\sum_{t\in S_i'}w_t$. Using the previous inequalities we can lower bound his utility as follows:
\begin{align*}
U_i^\M(a_i^*,a_{-i}; v_i) =~& \left(1-\frac{1}{\alpha}\right)\sum_{t\in S_i'}w_t\\
\geq~&\left(1-\frac{1}{\alpha}\right)\sum_{t\in S_i^*}w_t -\left(1-\frac{1}{\alpha}\right) \cdot \alpha\cdot\left(\sum_{t\in S\cap T_i }b_t + \sum_{t\in S\cap \ground_i}b_t\right)
\end{align*}
Summing over all players:
\begin{align*}
\sum_i U_i^\M(a_i^*,a_{-i}; v_i) \geq~&\left(1-\frac{1}{\alpha}\right)\opt(v) - \left(1-\frac{1}{\alpha}\right) \cdot \alpha\cdot\left(\sum_i \sum_{t\in S\cap T_i}b_t +\sum_{i}\sum_{t\in S\cap \ground_i} b_t\right)\\
\geq~& \left(1-\frac{1}{\alpha}\right)\opt(v)  -(\alpha-1)\cdot 2\sum_{t\in S}b_t
\end{align*}
where the last inequality follows, since $T_i$ are disjoint sets and thereby $\sum_{i} \sum_{t\in S\cap T_i} b_t\leq \sum_{t\in S}b_t$.
By setting $\alpha = \frac{1}{2}+1=\frac{3}{2}$, yields the result.
\end{proof}

\subsection{Action Space Restrictions and Extension to Polymatroids}
\label{sec.polymatroid}

We examine the generalization of Theorem \ref{thm.xos} to polymatroids. In a polymatroid setting, each element $t\in \ground_i$ corresponds to a divisible good. The allocation $\X_i = \R_+^{|\ground_i|}$ of a player is the vector of allocated units from each element $t\in \ground_i$: $x_i=(x_t)_{t\in \ground_i}$.  The mechanism chooses a vector of allocated units of each element: $x=(x_t)_{t\in \ground}\in \R_+^{|\ground|}$. This vector has to satisfy a polymatroid constraint: for any $S\subseteq \ground$, $\sum_{t\in S}x_t \leq f(S)$, and $f(\cdot)$ is a monotone submodular function with $f(\emptyset)=0$. 

The valuation of a player is linear across elements and is homogeneous for each element, i.e. $v_i(x_i) = \sum_{t\in \ground_i} w_t\cdot x_t$. We show that the greedy mechanism, described in Mechanism \ref{mech:polymatroid}, is smooth.
\renewcommand{\algorithmcfname}{Mechanism}
\begin{algorithm*}[h]\label{mech:polymatroid}
\SetKwInOut{Input}{Input}\SetKwInOut{Output}{Output}
\BlankLine
\nl From each player $i$ solicit bids $b_t$ for each $t\in \ground_i$. Denote with $a_i=(b_t)_{t\in \ground_i}$ and $b=(b_t)_{t\in \ground}$\\
\nl Run the greedy polymatroid algorithm with weights $b$ to decide the final allocation $x$, i.e. at each iteration pick element $t$ from remaining with maximum $b_t$ and increase $x_t$ until some polymatroid constraint becomes tight. Then remove $t$ from consideration.\\
\nl Charge each player $i$, $\sum_{t\in \ground_i} b_t\cdot x_t$.
\caption{Polymatroid mechanism.}
\label{mech:position-per-impression}
\end{algorithm*}
\renewcommand{\algorithmcfname}{ALGORITHM}

\begin{theorem}\label{thm:polymatroid}
The polymatroid mechanism is $\left(\frac{1}{3},1\right)$-smooth when players have linear and homogeneous valuations. 
(proof in Appendix \ref{sec:app-matroids})
\end{theorem}

The proof of Theorem \ref{thm:polymatroid} is based on the observation that, by discretizing the allocation space, we can view the polymatroid mechanism as the limit of a greedy matroid mechanism, where the players are restricted to submit the same bid on all units derived from the discretization of a good. Then the theorem follows essentially from Theorem \ref{thm:poa-greedy} by making the following generic observation on smooth mechanisms: if we restrict the action space of a $(\lambda,\mu)$-smooth mechanism, such that the smoothness deviations fall within the restricted action space, then the restricted mechanism is still $(\lambda,\mu)$-smooth. Since the smoothness deviation of Theorem \ref{thm:poa-greedy} is simply a scaled version of a player's true value, and the value of a player is homogeneous for each good, we get that this deviation remains in the restricted space of the discretized matroid mechanism, proving smoothness of each discretized version and thereby of the limit polymatroid mechanism. 

\paragraph{Submodular valuations} Suppose that each player's value $v_i(x_i)$ is a monotone submodular function on the euclidean lattice defined on $\R_+^{|\ground_i|}$. Then from \cite{Syrgkanis2013}, we know that it can be expressed as: 
$v_i(x_i) = \max_{\ell \in \El} \sum_{t\in \ground_i} v_t(x_t)$,
with $v_t(\cdot)$ being an increasing concave function. Since, it is easy to see that every increasing concave function can be expressed as the maximum of functions that are linear up to a point and then constant: i.e. $v_t(x_t) =\max_{\ell \in \El} w_t^{\ell}\min\{ x_t,  q_t^{\ell}\}$, we can conclude that any submodular valuation can be written as
$v_i(x_i) = \max_{\ell \in \El} \sum_{t\in \ground_i}  w_t^{\ell}\cdot \min\{x_t, q_t^{\ell}\}$.
for some index set $\El$. Thus in order to prove smoothness of the polymatroid mechanism, by Lemma \ref{lem:sum-to-max} it suffices to show smoothness for the following much simpler class of valuations:
$v_i(x_i) = \sum_{t\in \ground_i}  w_t\cdot \min\{x_t, q_t\}$.


To render the polymatroid mechanism smooth,
we introduce the following modification of the polymatroid mechanism under which the player can also submit allocation capacities for each good.  

\renewcommand{\algorithmcfname}{Mechanism}
\begin{algorithm*}[h]\label{mech:polymatroid-caps}
\SetKwInOut{Input}{Input}\SetKwInOut{Output}{Output}
\BlankLine
\nl From each player $i$ solicit a bid $b_t$ and a capacity $q_t$ for each $t\in \ground_i$. Denote with $a_i=(b_t)_{t\in \ground_i}$ and $b=(b_t)_{t\in \ground}$ and $q=(q_t)_{t\in \ground}$\\
\nl Run the greedy polymatroid algorithm with weights $b$ and capacities $q_t$ to decide the final allocation $x$, i.e. at each iteration pick element $t$ from remaining with maximum $b_t$ and increase $x_t$ until some polymatroid constraint becomes tight or $x_t$ reaches $q_t$. Then remove $t$ from consideration.\\
\nl Charge each player $i$, $\sum_{t\in \ground_i} b_t\cdot x_t$.
\caption{Polymatroid mechanism with capacities.}
\label{mech:position-per-impression2}
\end{algorithm*}
\renewcommand{\algorithmcfname}{ALGORITHM}


\begin{corollary} 
\label{cor.polymatroid.2}
The polymatroid mechanism with capacities is $\left(\frac{1}{3},1\right)$-smooth when players have submodular valuations.
(proof in Appendix \ref{sec:app-matroids})
\end{corollary}

\section{Applications}
\label{sec:applications}

We start by observing that the mechanism design setting with matroid feasibility constraints is the same as the setting analyzed by Bikhchandani et al. \cite{Bikhchandani2011} and therefore it can be applied to the plethora of applications identified there. These include: scheduling matroids of Demange et al. \cite{Demange1986}, uniform matroids which can capture the sale of a homogeneous good to bidders with decreasing marginal values (see e.g. the discriminatory analog of the uniform price auction studied in \cite{Syrgkanis2013}), transversal matroids that can capture pairwise Kidney exchanges \cite{Roth2005}. Additionally, the polymatroid result has applications in the spatially distributed markets of Babaioff et al. \cite{Babaioff2009}. For a more detailed exposition of these applications see Section 5 of \cite{Bikhchandani2011}.

In all these settings, our results shows that the very simple greedy first price mechanism achieves at least $1/3$ of the optimal allocation.  Note that when these examples are discussed in \cite{Bikhchandani2011}, the valuations of players are assumed to be linear; our results apply to these allocations but allow more generally for agents to have XOS valuations over the elements for sale, as in the scheduling example described at the beginning of the paper.

\paragraph{Simultaneous Item Auctions} Another application of our setting is that of the simultaneous first price item-bidding auction. The simultaneous item bidding auction among $m$ items can be seen as matroid mechanism on a matroid of $n\cdot m$  elements, where each element corresponds to a (player, item) pair.  Each player controls the $m$ elements in which he appears, one for each of the items. The feasibility constraint is that, for every item, only one of the elements containing that item can be in the outcome set, which is a matroid feasibility constraint.  Note that the simultaneous first price item-bidding mechanism corresponds to the greedy first price mechanism on this matroid setting.

\paragraph{Ad Auctions Polytope} Our polymatroid result finds interesting application in the case of ad auctions. By Goel et al \cite{Goel2012} the following setting can be represented as a polymatroid: there are multiple keywords and each keyword $j$ has $k$ slots available. Each slot $(j,k)$ is associated with a number of clicks $a_{j,k}$. A set of players is competing for the clicks. Each player has a per-click value $v_i$ and his value is $v_i\cdot x_i$, where $x_i$ is the number of clicks he got, i.e. if $k(i,j)$ is the position he won at keyword $j$, then $x_i=\sum_{j} k_{j,k(i,j)}$. The feasibility constraints in the above setting can be expressed as a polymatroid and therefore the greedy polymatroid mechanism yields a smooth mechanism. Note that the greedy polymatroid mechanism in this setting is to run independent Generalized First Price auctions for each keyword. 
%
In fact, our mechanism can handle the case where players have different valuations for different keywords, i.e. $\sum_{j} v_{ij} a_{j, k(i,j)}$. In this case, each keyword is a different coordinate of the polymatroid.  The mechanism again corresponds to running a Generalized First Price auction, but now the players can submit a bid $b_{ij}$ for each keyword $j$.

Our polymatroid mechanism with capacities can also have applications if the valuations of the players are submodular over the allocated clicks from each keyword. However, the mechanism would not correspond to running GFP on each keyword.  Instead, the mechanism would asks also for capacities of the players and then randomize over the allocation of positions so as in expectation a player is never allocated more than his declared click capacity.

\paragraph{Ad Auctions with Externalities}
As an illustration of the power of the polymatroid formulation, we revisit a model of sponsored search ad auctions with externalities studied by Roughgarden and Tardos \cite{RT12}.  Like the Ad Auction setting above, there are $k$ slots; for simplicity assume there is only one keyword.  Each advertiser has a per-click value $v_i$, but also a click probability $p_i$ and a continuation probability $q_i$.  Rather than presuming that each slot has a click probability, these probabilities arise as follows: given an assignment of ads to slots, a user will begin scanning the list of ads from slot $1$ to slot $k$.  Upon encountering an ad $i$, say in slot $j$, it will click the ad with probability $p_i$.  Also, with probability $q_i$ the user will continue scanning to the next slot; otherwise, with probability $1 - q_i$ the user will terminate and not view any subsequent ads.  Thus, if we write $\alpha(j)$ for the ad assigned to slot $j$, then the total value of an ad $i$ assigned to slot $\ell$ is $v_i p_i \prod_{j < \ell} q_{\alpha(j)}$.

In \cite{RT12} it is shown that an auction that assigns ads to slots greedily by $\frac{v_i p_i}{1 - q_i}$ and charges threshold \vsedit{payments has a price of anarchy of at most 4.  We will describe how to interpret this setting as a polymatroid auction, and our results will then imply a price of anarchy of $3$ for a first price variant. Our result price of anarchy of $3$ for polymatroid auctions, though phrased for first price payments, can also be easily adapted to give a price of anarchy of at most $6$ for any second price variant, assuming players do not bid above their valuation.\footnote{A fairly standard, in the literature, modification in the proof of our main theorem to show that the auction is weakly $(1/3,1)$-smooth as defined in \cite{Syrgkanis2013}, which implies a bound of $6$ on the price of anarchy} Thus our polymatroid result directly implies a constant approximation for the auction analysed in \cite{RT12}.}  

Denote with $x_i$ the probability that ad $i$ is the last ad scanned.  Then advertiser $i$ obtains expected value $\frac{v_i p_i}{1 - q_i} \cdot x_i$.  We can therefore write $V_i = \frac{v_i p_i}{1 - q_i}$, and then view the auction as allocating the probabilities $x_i$ to bidders with modified values $V_i$.  We now claim that the set allocation profiles $\mathbf{x}$ corresponding to valid assignments forms a polymatroid, with constraints given by $\sum_{t \in S}x_t \leq 1 - \prod_{t \in S}q_t$.  To see that the conditions are necessary, note that the probability that any ad in set $S$ is the last scanned is maximized when the ads in $S$ occupy the lowest-indexed slots, in which case the probability is precisely $1 - \prod_{t \in S}q_t$.  For sufficiency, note that any deterministic assignment certainly satisfies these constraints (consider all prefixes of the assignment), and hence any mixture of assignments would as well. \vsedit{Last observe that the function $f(S) = 1- \prod_{t \in S}q_t$ is a monotone submodular function and therefore the polytope defined by these constraints is a polymatroid.}

Since the setting can be framed as a polymatroid auction with values given by $\frac{v_i p_i}{1 - q_i}$, we can conclude that a greedy auction that allocates in order of $\frac{v_i p_i}{1 - q_i}$ is $(\frac{1}{3}, 1)$-smooth.



\section{Smoothness for Matroid Intersections}

We now turn to more complex feasibility constraints. We start by showing that the greedy algorithm leads to a smooth mechanism via swap deviations, for constant $\lambda$ and $\mu$, even when the feasibility constraint corresponds to a matching constraint: i.e. each ground element $t\in \ground$ corresponds to an edge $(u_t,v_t)$ in a bipartite graph $(U,V,\ground)$ which is an intersection of two partition matroids (\cite{Schrijver2003}). 

Then we turn to the intersection of arbitrary $k$ matroids and we show that running the optimal algorithm (which in general is an NP-hard problem) rather than the greedy algorithm yields a smooth mechanism, whose parameters $\lambda$ and $\mu$ degrade with $k$. We conjecture that similar behavior holds for the greedy algorithm, but we leave it as an open question for future research.  
\begin{conj} The $\poa$ is $O(k)$ for the greedy mechanism when the feasibility constraint is the intersection of $k$ matroids.  
\end{conj}

\subsection{Matchings and Greedy Allocation}
\label{sec.matchings}

We present the smoothness theorem for matching feasibility constraints. Unlike the matroid setting where we used the existing machinery of Generalized Rota exchanges to create a charging argument, in this setting there is no analogous machinery for the greedy algorithm. Hence, we construct a charging argument that allows us to show that from a deviation either a player gets high utility or some part of the revenue at equilibrium is high. Moreover, each part of the equilibrium revenue is not charged more than twice by our charging scheme.

\begin{figure}[t]
\centering
\subfigure[example with a matching constraint.]
{\label{f:matching}
{
\includegraphics[width=2.7in]{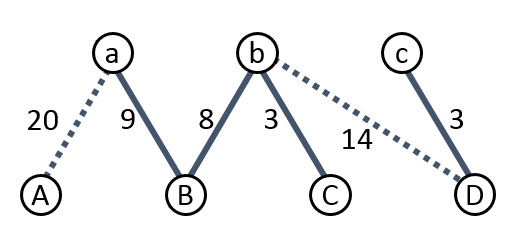}
}
}
\caption{Sample input instance to a greedy mechanisms.  Elements are edges, and weights represent agent values.  Each solid edge is being bid on truthfully, each dashed edge has a bid of $0$.  There is a matching constraint: no edges that share an edge can be simultaneously allocated.}
\label{fig:example2}
\end{figure}

Let us spend a moment to illustrate the difficulty inherent in constructing such mappings, 
Consider the matching instance described in Figure \ref{f:matching}, where the elements are the edges of the graph and the feasibility constraint is that no two allocated edges can share a vertex.  Each edge is labeled with a value, and we consider a bid profile where the bids on dashed edges are set to $0$.  The optimal outcome is $\{aA,bD\}$, and the outcome of the bid profile is $\{aB,bC,cD\}$.  Suppose a single agent, agent $1$, is the one bidding on edges $aA$ and $bD$.  Say agent $1$ were to switch to declaring half of his value for each of those two edges.  Then he would win edge $aA$, but the next edge chosen by the greedy algorithm would be $bB$, and edge $bD$ would not be chosen.  To which edge should the loss of edge $bD$ be charged?  It cannot be edge $bB$, since that edge isn't present in the equilibrium outcome, so its payment is not part of the equilibrium payment.  It also cannot be $bC$ or $cD$, since their payments are not large enough.  In this example, one could charge against edge $aB$, even though that edge doesn't share a vertex with $bD$.  The heart of our analysis of greedy algorithms under a matching constraint is the design of a ``long-range'' charging scheme that generalizes the scenario described in this example.  

We ultimately obtain the following smoothness condition on greedy mechanisms subject to matching constraints.

\begin{theorem}\label{thm:poa-greedy-matching}
The greedy mechanism is $\left(\frac{1}{2},4\right)$-smooth via swap deviations when valuations are additive and the feasibility constraint is a matching constraint.
\end{theorem}
\begin{proof}
Consider a valuation profile $v$ and an action profile $a$ with corresponding bid profile $b$ and let $S$ be the set output by the greedy algorithm on action profile $a$. Suppose that each player $i$ deviates to bidding the pointwise maximum of $a_i$ and $a_i'=\left(\frac{w_t}{2}\right)_{t\in S^*_i}$. This is a valid swap deviation according to definition \ref{defn:swap-smooth}.  Let $a'=(a_i',a_{-i})$ be the action profile when player $i$ deviates and $b'$ the corresponding bid profile.  Let $S'$ be the set allocated after the deviation and let $S'_i=S' \cap \ground_i$.  

Write $A_i = S^*_i \cap S'_i$ and $U_i = S^*_i - S'_i$.  (So $A_i$ is the set of optimal items for agent $i$ that are allocated under $S'$, and $U_i$ are those that are unallocated). 
%
We have:
\[U_i^\M(a_i', a_{-i};v_i) \geq \sum_{t\in A_i}\left( w_t -\max\left\{\frac{w_t}{2},b_t\right\}\right)\geq \frac{1}{2} \sum_{t \in A_i} w_t- \sum_{t\in A_i}b_t\geq \frac{1}{2} \sum_{t \in A_i} w_t- \sum_{t\in S_i^*}b_t.\]  
So 
\[\sum_i U_i^\M(a_i', a_{-i};v_i) \geq \frac{1}{2} \sum_i \sum_{t \in A_i} w_t-\sum_i\sum_{t\in S_i^*}b_t=\frac{1}{2} \sum_i \sum_{t \in A_i} w_t-\sum_{t\in S^*}b_t.\] 
Since $S^*$ is a feasible allocation and $S$ is the outcome of the greedy algorithm on matchings, then by the well known $2$-approximation guarantee of the algorithm we have $\sum_{t\in S^*}b_t\leq 2\sum_{t\in S}b_t$. Combining with the inequality above, this yields
$\sum_i U_i^\M(a_i', a_{-i};v_i) \geq \frac{1}{2} \sum_i \sum_{t \in A_i} w_t-2\sum_{t\in S}b_t$.

We claim (proved below) that there exists a mapping $\phi \colon (\cup_i U_i) \to S$ such that $b_{\phi(t)} \geq \frac{1}{2}w_t$ for all $t \in (\cup_i U_i)$, and moreover $|\phi^{-1}(x)| \leq 2$ for all $x \in S^g$.  This will imply that $2 \sum_{x \in S}b_x \geq \frac{1}{2}\sum_i \sum_{t \in U_i} w_t$.  We can then establish the smoothness property:
\begin{align*} \sum_i U_i^\M(a_i', a_{-i};v_i) \geq~& \frac{1}{2} \sum_i \sum_{t \in A_i} w_t  -2 \sum_{x\in S}b_x+ \frac{1}{2} \sum_i \sum_{t \in U_i} w_t - 2 \sum_{x \in S}b_x \\
=~& \frac{1}{2} \opt(v) - 4 \sum_{x \in S}b_x
\end{align*}

\paragraph{Construction of charging mapping.} It remains to construct the promised mapping $\phi$. We first introduce the notion of an exchange graph for sets that are in the intersection of two matroids  $\matroid_1 = (\ground, I_1)$ and $\matroid_2 = (\ground, I_2)$:
\begin{defn}[Exchange Graph]\label{defn:exchange-graph} For a set $S \subset I_1 \cap I_2$, the exchange graph for $S$ is a directed bipartite graph $G(S)$ with node sets $S$ and $\ground-S$ such that: for $v \in S$ and $u \in \ground\backslash S$, we have edge $(v,u)$ if $S - v + u \in I_1$, and edge $(u,v)$ if $S - v + u \in I_2$.
\end{defn}

In a matching, a set of elements is in an independent set of matroid $\matroid_1$ if no two elements have the same left endpoint in the bipartite graph $(U,V,\ground)$, while it is an independent set of matroid $\matroid_2$, if no two elements have the same right endpoint. Then a feasible set of the mechanism can be viewed as the intersection of these two matroids. We now provide some extra properties of the exchange graph in the special case of a matching feasibility constraint. 

\begin{observation}\label{obs:T-to-S-uniqueness}
Given $S \in I_1 \cap I_2$ and $t \in \ground - S$, there is at most one $s \in S$ such that $(s,t) \in G(S)$ and at most one $s' \in S$ such that $(t,s') \in G(S)$.
\end{observation}

\begin{observation}\label{obs:S-to-T-uniqueness}
Given $S \in I_1 \cap I_2$, $T \in I_1 \cap I_2$, and $s \in S$, there is at most one $t \in T - S$ such that $(s,t) \in G(S)$ and at most one $t' \in T-S$ such that $(t',s) \in G(S)$.
\end{observation}

Next we argue about the structure of $S'$, by way of $G(S)$. We remind that $S$ is the greedy outcome under bid profile $b$ and $S'$ is the greedy outcome under bid profile $b'$ produced by action profile $a'=(a_i',a_{-i})$. We also remind that the bids in $b'$ are the same as the bids in $b$ with only some elements $T \in I_1 \cap I_2$ having an increased bid. 

\begin{lemma}\label{lem:matching-path-structure}
There exist paths $\pi_1, \dotsc, \pi_\ell$ in $G(S)$ such that:
\begin{enumerate}
\itemsep-0.3em
\item in each path, the bids $b_t'$ on the nodes are either monotonically increasing or decreasing,
\item in each path, the node with maximum bid from $b'$ is the unique element of $T-S$ on the path,
\item the paths are disjoint, except that each $t \in T$ could be the maximum-$b'$ element for at most one increasing (in $b'$) path and one decreasing (in $b'$) path, and
\item $S'$ is precisely $S$ with all nodes from $S \cap (\cup_i \pi_i)$ removed and all nodes from $S - (\cup_i \pi_i)$ added.
\end{enumerate}
\end{lemma}
\begin{proof}
Let $H_k$ be the top $k$ elements from $\ground$ with respect to bids $B'$ (breaking ties in the same manner as the greedy algorithm).  We will prove the stronger result that, for each $k$, our lemma holds for the sets $S \cap H_k$ and $S' \cap H_k$ (in item number 4).  The proof will be by induction on $k$.  Taking $k = |\ground|$ will then give the stated lemma.

For the base case $k=1$, let $x$ be the single element in $H_1$.  If $x \not\in T$ then $b_x = b_x'$ and hence $S \cap H_1 = \{x\} = S' \cap H_1$ so the result holds trivially.  If $x \in T$ and $x \in S$ then the result again holds trivially.  If $x \in T$ but $x \not\in S$, then we have $S \cap H_1 = \emptyset$ but $S' \cap H_1 = \{x\}$.  In this case, take path $P_1$ to be the singleton node $\{x\}$ to get the desired result.

Now consider $k > 1$.  By induction, there are paths $\pi_1, \dotsc, \pi_\ell$ with the required properties for $S \cap H_{k-1}$ and $S' \cap H_{k-1}$.  Let $x$ be the single element in $H_k - H_{k-1}$.  If $x$ is in both $S$ and $S'$, or in neither, then paths $\pi_1, \dotsc, \pi_\ell$ satisfy the required properties.  

Suppose $x \in S$ but $x \not\in S'$.  Then, since $x \not\in S'$, there exists some element $y \in S' - S$ such that $b_y' \geq b_x'$ and either $(x,y)$ or $(y,x)$ is in $G(S)$.  Assume $(x,y) \in G(S)$, as the other case is symmetric.  Since $y$ is considered before $x$ by the greedy algorithm on $a'$, we have $y \in H_{k-1}$.  So $y$ must lie on a path $\pi_i$.  From our earlier observation, there can be no $x' \in S$, $x' \neq x$, such that $(x',y) \in G(S)$.  Thus, either $y$ is the maximum-$a'$ element of $\pi_i$, $\pi_i$ is decreasing and no increasing path ending at $y$, or else $y$ is the endpoint of $\pi_i$ with lowest bid $a'$.  In either case, extending $\pi_i$ by appending $x$ retains the required properties of our paths.

Finally, suppose $x \not\in S$ but $x \in S'$.  If $x \in T$ then create a new path containing only the singleton $x$, and we are done.  Otherwise, there must be some element $y \in S - S'$ such that $b_y \geq b_x$ (and hence $b_y' \geq b_x'$, since $x \not\in T$) and either $(x,y)$ or $(y,x)$ is in $G(S)$.  Assume $(x,y) \in G(S)$, as the other case is symmetric.  This case is now similar to the previous case.  Since $y$ is considered before $x$ by the greedy algorithm on $a$, and hence on $a'$, we have $y \in H_{k-1}$.  So $y$ must lie on a path $\pi_i$.  From our earlier observation, there can be no $x' \in S'$, $x' \neq x$, such that $(x',y) \in G(S)$ (since $S'$ is in $I_1 \cap I_2$).  Thus $y$ is an endpoint of a path $\pi_i$. Since it cannot be the maximum-$a'$ endpoint of the path, it is the endpoint of $\pi_i$ with lowest-$a'$ bid.  In either case, extending $\pi_i$ by appending $x$ retains the required properties of our paths.

So, in all cases, the required paths exist for this value of $k$.  The result follows by induction.
\end{proof}

Finally, we argue about properties of elements not allocated by the greedy algorithm when their bids are increased.
As before, fix a bid profile $b$ and greedy outcome $S$ for $b$, and suppose bid profile $b'$ is $b$ with increased bids on some set of elements $T \in I_1 \cap I_2$.  Let $S'$ be the greedy outcome for $b'$.

\begin{lemma}\label{lem:matching-charging}
\label{lem.charge}
Suppose $t \in T - S'$.  Then $t$ is adjacent (in $G(S)$) to some $x \in S$ such that either 
\begin{enumerate}
\itemsep-0.3em
\item\label{cond:first} $b_x \geq b_t'$, or 
\item\label{cond:sec} $x \not\in S'$ and $x$ lies on a path $\pi_i$ (from Lemma \ref{lem:matching-path-structure}) with a neighbor $y\not\in T$ such that $b_y' \geq b_t'$, or
\item \label{cond:third} $x \not\in S'$ and $x$ has a neighbor $y\not\in T$ in $G(S)$ such that $y$ is the endpoint of a path $\pi_i$ (from the previous lemma) and $b_y' \geq b_t'$. Moreover, the path $\pi_i$ is increasing if $(x,t)\in G(S)$ and decreasing if $(t,x)\in G(S)$.
\end{enumerate}
\end{lemma}
\begin{proof}
If $t=(\alpha,\beta)\not\in S'$, there must exist some $y \in S'$ with $b_z' \geq b_t'$ and $y$ conflicting (i.e., shares a vertex with) with $t$ (wlog let it be vertex $\alpha$).  If $y \in S$ then take $x = y$ and we're done.  Let's assume that this is not the case and Condition \ref{cond:first} is not true.

Then, we have $y \in S' - S$, so our previous lemma states that $z$ lies on a path $\pi_i$ with the appropriate properties.  Since $y$ shares vertex $\alpha$ with $t$, we can take $x$ to be the element of $S$ that also shares this vertex.  Then $(x,y)$ and $(x,t)$ are in $G(S)$, and moreover $x \not\in S'$. If $y$ is the endpoint of path $\pi_i$ then we are done, since Condition \ref{cond:third} is satisfied. 

Otherwise, it must be that the path $\pi_i$ continues. We need to argue that the path is increasing and that $(x,y)$ is part of the path. If the path was decreasing then it means that there exists some edge $(x',y)$ such that $b_{x'}'\geq b_y'$. Thus $x'$ shares vertex $\alpha$ with $y$ and hence with $t$. Moreover, for this reason $x'\not\in T$ and therefore $b_x'=b_x$. Thus condition \ref{cond:first} is satisfied, with $x=x'$, a contradiction to our first assumption. So it must be that the path is increasing and hence an edge $(r,y)$ is part of the path. But from Observation \ref{obs:T-to-S-uniqueness} the only such $r$ is  $x$. Thus Condition \ref{cond:sec} is satisfied.
\end{proof}

We're now ready to define our promised mapping $\phi: \cup_i U_i\rightarrow S$.  Take $T = S^*_i$ in the above Lemmas.  Note $U_i = T - S'$.  For each $t \in T - S'$, if Condition \ref{cond:first} of Lemma \ref{lem.charge} is satisfied then take $\phi(t)=x$, for the $x$ in the condition.
If Conditions \ref{cond:sec} or \ref{cond:third} of Lemma \ref{lem.charge} are satisfied then we first create a temporary association, based on which we subsequently construct the mapping. If Condition \ref{cond:sec} of  Lemma \ref{lem.charge} is satisfied then we associate $t$ with the $x$ in the condition, whilst if Condition \ref{cond:third} is satisfied then we associate $t$ with the endpoint $y$. Then each $x \in S$ is associated with only one $t$, since each such $x$ lies on a unique path, and the neighbor from $T$ with which it's associated is determined by the direction of that path. Moreover, each endpoint of a path $y$ is associated with only one other node $t$, since $y$ lies on a unique path and the node $t$ with which it is associated is determined by the monotonicity of the path.

Now define $\phi(t)$ as follows: starting from the node $x$ associated with $t$, follow the path in the direction of increasing $b'$ until reaching some $x' \in S$ that is either (a) associated with some other node $t' \in T$, or (b) the last element of $S$ along the path.  In either case, we will set $\phi(t) = x'$.  Note that $b_{x'}' \geq b_t'$, since by construction $b_{x'}' \geq b_y' \geq b_t'$. Moreover, $x'\not\in T$ by the observation that $x'$ is adjacent to (and hence conflicts with) some $x\in T$ in graph $G(S)$. 
Thus $b_x=b_x'\geq b_t'$ as required.

We must argue that this mapping $\phi$ satisfies the required properties.
We already have that $b_{\phi(t)} \geq b_t'\geq \frac{w_t}{2}$ for each $t$.  We next argue that $|\phi^{-1}(x)| \leq 2$ for each $x \in S$.  This follows because each $x$ is mapped-to at most once for each of its (two) adjacent elements $t \in S^*$.  If $b_x$ is greater than $b_t'$, then $x$ is mapped-to directly from $t$.  If $b_x$ is less than $b_t'$ then $x$ can potentially be mapped to via an association with $t$ along the (at most one) path containing $x$, but only by one other element $t'$ (i.e., corresponding to the next-lowest element along that path that has an association). Hence, the mapping $\phi$ satisfies the required properties, completing the proof of Theorem \ref{thm:poa-greedy-matching}.
\end{proof}
\subsection{Intersections of Matroids and Optimal Algorithm}

We now analyze the mechanism where instead of running the greedy algorithm over the reported bid profile, we run the optimal algorithm to decide the outcome set. Then each player is charged his bid for his allocated set. We refer to this mechanism as the optimal mechanism with first prices. 

\begin{theorem}\label{thm:opt-intersection}
The optimal mechanism with first prices  is a $\left(\frac{1}{k+2},1\right)$-smooth mechanism when valuations are $\xos$ and the feasibility constraint on the ground set is the intersection of $k$ matroids. 
\end{theorem}
\begin{proof}
By Lemma \ref{lem:sum-to-max} it suffices to prove the theorem for additive valuations. Consider an additive valuation profile $v$. Suppose that each player $i$ deviates to $a_i^*=\left(\frac{w_t}{\alpha}\right)_{t\in \ground_i}$. Let $S^*$ be the optimal base for valuation profile $v$ and
$S_i^*=S^*\cap \ground_i$, be player $i$'s allocation in the optimal base. 

Consider an action profile $a$, where $a_i=(b_t)_{t\in \ground_i}$, and let $S$, be the selected set under action profile $a$. Let $a'=(a_i^*,a_{-i})$, be the induced action profile and $S'$ be the set allocated after the deviation and $S_i'=S'\cap \ground_i$. 

Suppose that the feasibility constraint on the elements of the ground set is the intersection of $k$ matroid constraints, $\matroid_1,\ldots,\matroid_k$. By Lemma \ref{lem:exchange} applied to every matroid $\matroid_t=\left(\ground,\I^t\right)$, we have that there exist disjoint sets $T_1^t,\ldots,T_n^t$ such that $S_i^*\cup (S- T_i^t)\in \I^t$. Thus it is easy to see that: $Q=S_i^*\cup \left(S- \cup_{t=1}^{k} T_i^t\right) \in \cap_{t=1}^{k}\I^t$ is a feasible set. Let $T_i=\cup_{t=1}^{k} T_i^t$, observe that 
since for each $t\in [k]$, $T_1^t,\ldots,T_n^t$, are disjoint sets, an element appears in at most $k$ of the sets $T_1,\ldots,T_n$.

The rest of the proof follows along similar lines as in theorem \ref{thm:poa-greedy}. By the optimality of the algorithm on the reported bid profile and since $Q=S_i^*\cup \left(S- T_i\right) $ is feasible, we have: 
\begin{align*}
 \sum_{t\in S_i'}\frac{w_t}{\alpha} + \sum_{t\in S'-S_i'}b_t\geq~& \sum_{t\in S_i^*}\frac{w_t}{\alpha} + \sum_{t\in S-T_i- \ground_i}b_t + \sum_{t\in (S\cap\ground_i) -T_i}\frac{w_t}{\alpha}\\
\geq~& \sum_{t\in S_i^*}\frac{w_t}{\alpha} + \sum_{t\in S-T_i-\ground_i}b_t
\end{align*}
By optimality of the algorithm on the initial bid profile we have:
\begin{equation*}
\sum_{t\in S}b_t=W(S,a) \geq W(S',a) = \sum_{t\in S'}b_t \geq \sum_{t\in S'-S_i'}b_t
\end{equation*}
Combining we get:
\begin{align*}
 \sum_{t\in S_i'}\frac{w_t}{\alpha} \geq~& \sum_{t\in S_i^*}\frac{w_t}{\alpha} + \sum_{t\in S-T_i- \ground_i}b_t -  \sum_{t\in S}b_t\\
=~&\sum_{t\in S_i^*}\frac{w_t}{\alpha} - \sum_{t\in S\cap (T_i \cup \ground_i)}b_t \\
\geq~&\sum_{t\in S_i^*}\frac{w_t}{\alpha} - \sum_{t\in S\cap T_i }b_t - \sum_{t\in S\cap \ground_i}b_t
\end{align*}
Observe that by definition the utility of the player under the deviation is: $U_i^\M(a_i^*,a_{-i}; v_i) = \left(1-\frac{1}{\alpha}\right)\sum_{t\in S_i'}w_t$. Using the previous inequalities we can lower bound his utility as follows:
\begin{align*}
U_i^\M(a_i^*,a_{-i}; v_i)=~& \left(1-\frac{1}{\alpha}\right)\sum_{t\in S_i'}w_t\\
\geq~& \left(1-\frac{1}{\alpha}\right)\sum_{t\in S_i^*}w_t -\left(1-\frac{1}{\alpha}\right) \cdot \alpha\cdot\left(\sum_{t\in S\cap T_i }b_t + \sum_{t\in S\cap \ground_i}b_t\right)
\end{align*}
Summing over all players:
\begin{align*}
\sum_i U_i^\M(a_i^*,a_{-i}; v_i) \geq~&\left(1-\frac{1}{\alpha}\right)\opt(v) - \left(1-\frac{1}{\alpha}\right) \cdot \alpha\cdot\left(\sum_i \sum_{t\in S\cap T_i}b_t +\sum_{i}\sum_{t\in S\cap \ground_i} b_t\right)\\
\geq~& \left(1-\frac{1}{\alpha}\right)\opt(v)  -(\alpha-1)\cdot (k+1)\sum_{t\in S}b_t
\end{align*}
where the last inequality follows, since an element $t\in \ground$ appears in at most $k$ of the sets $T_1,\ldots,T_n$ and thereby $\sum_{i} \sum_{t\in S\cap T_i} b_t\leq k\sum_{t\in S}b_t$.
By setting $\alpha = \frac{1}{k+1}+1=\frac{k+2}{k+1}$, yields the result.
\end{proof}

\bibliographystyle{acmsmall}
\bibliography{thesis-bib}

\begin{appendix}

\section{Omitted Proofs}

\subsection{Polymatroids}\label{sec:app-matroids}

\begin{proofof}{Theorem \ref{thm:polymatroid}}
It is well-known (see e.g. Bikhchandani et al. \cite{Bikhchandani2011} or Schrijver \cite{Schrijver2003}) that for a sufficiently small discretization of the allocation space in $\delta$ units, if we consider the extended ground set where each element $t$, is duplicated $\frac{f(\{t\})}{\delta}$ times, then the feasibility constraint implied by the polymatroid on these extended element set is a matroid. Moreover, if we denote with $(t,k)$ the $k$-th copy of element $t$, then if we assign a weight of $b_t\cdot \delta$ to each element $(t,k)$, then as the discretization goes to zero, the greedy algorithm on the matroid corresponds to the greedy algorithm on the polymatroid. Subsequently the payment of the polymatroid mechanism coincides with the payment of the extended matroid mechanism, when run on bids $b_t\cdot \delta$ for each copy of $t$. Last if we denote with $w_t'=w_t\cdot \delta$, then the value of a player for an allocation of discretized units, corresponds to the value of a player in the discretized matroid that has value $w_t'$ for each copy of element $t$.

Thus we can view the polymatroid mechanism as the limit of a matroid mechanism where the players are restricted to submit the same bid on all copies of the same element. If we show smoothness of this restricted bid mechanism, then the smoothness of the polymatroid mechanism will follow by taking the limit of the discretization to zero.

In order to show smoothness, the only thing we need to observe is that the deviations used in the smoothness proof of the matroid mechanism (Lemma \ref{thm:poa-greedy}) are simply scaled versions of the valuation of a player. Thus it is easy to see, that if such scaled versions are allowed in the restricted bid space, then the same proof shows smoothness of the matroid mechanism under the restricted bid space. This is formalized in the following observation.
\begin{observation}
Suppose that a mechanism $\M$ is $(\lambda,\mu)$-smooth. Consider the mechanism $\M'$, which is identical to $\M$, with the exception that the action space of each player is restricted to some subset $\A_i'\subset \A_i$. If every action in the support of the deviations $\al_i^*(v)$ used to show smoothness of $\M$, fall into action space $\A_i'$, then $\M'$ is also $(\lambda,\mu)$-smooth.
\end{observation}

Since we assumed that the value of a player is additive and homogeneous, observe that the weight $w_t'$ of a player for each element of the discretized matroid is identical for all copies of element $t$. Thus the smoothness deviations of Lemma \ref{thm:poa-greedy} would correspond to bidding $\frac{w_t'}{\alpha}$ for each element of $t$, which is an action that belongs to the restricted strategy space. Thus the matroid mechanism is smooth even under this restricted space, as long as the value of a player for all copies of an element is identical. Hence, the theorem follows.
\end{proofof}

\begin{proofof}{Corollary \ref{cor.polymatroid.2}}
We first recall the submodular valuations setting.  Suppose that each player's value $v_i(x_i)$ is a monotone submodular function on the euclidean lattice defined on $\R_+^{|\ground_i|}$. Then from \cite{Syrgkanis2013}, we know that it can be expressed as: 
$v_i(x_i) = \max_{\ell \in \El} \sum_{t\in \ground_i} v_t(x_t)$,
with $v_t(\cdot)$ being an increasing concave function. Since, it is easy to see that every increasing concave function can be expressed as the maximum of functions that are linear up to a point and then constant: i.e. $v_t(x_t) =\max_{\ell \in \El} w_t^{\ell}\min\{ x_t,  q_t^{\ell}\}$, we can conclude that any submodular valuation can be written as
$v_i(x_i) = \max_{\ell \in \El} \sum_{t\in \ground_i}  w_t^{\ell}\cdot \min\{x_t, q_t^{\ell}\}$.
for some index set $\El$. Thus in order to prove smoothness of the polymatroid mechanism, by Lemma \ref{lem:sum-to-max} it suffices to show smoothness for the following much simpler class of valuations:
$v_i(x_i) = \sum_{t\in \ground_i}  w_t\cdot \min\{x_t, q_t\}$.

However, we readily observe that if we consider an arbitrarily small discretization of the polymatroid then the valuations of the players for the copies of an element $t$, will not be identical. Instead, if we consider some arbitrary order of the copies, then the player will have a value of $w_t\cdot \delta$ for the first $\frac{q_t}{\delta}$ copies and zero value for subsequent copies.
Thus to render the polymatroid mechanism smooth, we need to allow for the player to express such valuations for the copies of the same element. To achieve this 
we introduce the modification of the polymatroid mechanism under which the player can also submit allocation capacities for each good.

If we consider an arbitrarily small discretization of the polymatroid, then the player can submit a valuation that is $\frac{w_t\cdot \delta}{\alpha}$, for the first $\frac{q_t}{\delta}$ copies of element $t$ and zero for the remaining, by simply submitting a weight of $\frac{w_t}{\alpha}$ and a capacity of $q_t$, to the polymatroid mechanism with capacities.
\end{proofof}

\end{appendix}

\end{document}